\documentclass{article}
\usepackage{geometry}                
\usepackage{graphicx}
\usepackage{stmaryrd}
\usepackage{amssymb}
\usepackage{amsthm, bm}
\usepackage{amsmath, cancel, centernot }
\usepackage[mathscr]{eucal}
\usepackage{mathtools}
\usepackage{epstopdf}
\title{Bell Inequality Violation with Free Choice and Local Causality on the Invariant Set}

\author{T.N.Palmer\\ Department of Physics, University of Oxford, UK\\
tim.palmer@physics.ox.ac.uk}
\date{\today}                                          
\makeatletter
\newcommand\be{\@ifstar{\[}{\begin{equation}}}
\newcommand\ee{\@ifstar{\]}{\end{equation}}}
\newcommand\bp{\begin{pmatrix}}
\newcommand\ep{\end{pmatrix}}

\newtheorem*{theorem}{Theorem}
\newtheorem*{definition}{Definition}
\makeatother
\begin{document}
\bibliographystyle{plain}
\maketitle
\begin{abstract}
Bell's Theorem requires any theory which obeys the technical definitions of Free Choice and Local Causality to satisfy the Bell inequality. Invariant set theory is a finite theory of quantum physics which violates the Bell inequality exactly as does quantum theory: in it neither Free Choice nor Local Causality hold, consistent with Bell's Theorem. However, within the proposed theory, the mathematical expressions of both Free Choice and Local Causality involve states which, for number-theoretic reasons, cannot be ontic (cannot lie on the theory's fractal-like invariant set $I_U$ in state space). Weaker versions of Free Choice and Local Causality are proposed involving only the theory's ontic states. Conventional hidden-variable theories satisfying only these weaker definitions still obey the Bell inequality. However, invariant set theory, which violates the Bell inequality, satisfies these weaker definitions. It is argued that the weaker definitions are consistent with the physical meaning behind free choice and local causality as defined in space-time, and hence that Free Choice and Local Causality are physically too strong. It is concluded that the experimental violation of the Bell inequality may have less to do with free choice or local causality \emph{per se}, and more to do with the presence of a holistic but causal state-space geometry onto which quantum ontic states are constrained. \end{abstract}

\section{Introduction}
Bell's Theorem \cite{Bell} requires that any theory in which Free Choice and Local Causality hold must necessarily satisfy the CHSH version of the Bell inequality (and therefore be inconsistent with experiment). As a mathematical theorem, this result is clearly dependent on the precise mathematical definition of these two conditions, hence the reason for using capitalised words above - in lower case form they will merely refer to the more qualitative physical concepts underpinning such expressions. The purpose of this paper is to examine these definitions carefully in the context of a finite theory of qubit physics (`invariant set theory' or IST, introduced in Sections \ref{bell} and \ref{IST}) which violates the CHSH inequality exactly as does quantum theory. In this theory, the laws of physics derive from a primal fractal-like geometry $I_U$ in state space. It is shown explicitly that this theory does not obey Free Choice and Local Causality as it must do to be consistent with Bell's Theorem. However, as discussed in Section \ref{bellist} the mathematical expressions for Free Choice and Local Causality in IST necessarily involve states which (because of number theoretic properties of the cosine function) cannot lie on $I_U$ and therefore cannot be ontic. As such it is concluded that Free Choice and Local Causality are overly strong constraints for expressing the underpinning physical notions of free choice and local causality. In Section \ref{free}, Free Choice and Local Causality are weakened so that by construction they only refer to ontic states on $I_U$. These weakened definitions are defined as `Free Choice on the Invariant Set' and 'Local Causality on the Invariant Set'.  These weakened definitions have no impact on conventional hidden-variable theories, which continue to satisfy the Bell inequality. However, the ontic states of IST do satisfy the weakened definitions despite violating the Bell inequality as does quantum theory. In Section \ref{objections}, it is shown that invariant set theory is not of the `superdeterministic' type \cite{WisemanCavalcanti} involving implausible conspiracies between the implicit hidden variables of the theory and determinants of experimental settings. Neither is the theory fine tuned \cite{WoodSpekkens}. 

A principal conclusion of this paper is that the notions of free choice and local causality may ultimately not be central to a deep understanding of why quantum physics violates the Bell inequality and is therefore so different from classical physics. It is concluded it is the geometric structure of (cosmological) state space that instead lies at the heart of this search. As discussed in Section \ref{discussion}, a key motivation for this work is to find a formulation of quantum physics which is compatible with the nonlinear deterministic causal geometric structure of general relativity. As such, understanding the geometry of state space may also be the key to synthesising quantum and gravitational physics.   

\section{Bell's Theorem}
\label{bell}
Although Bell's Theorem can be formulated for inherently indeterministic theories of physics (e.g. \cite{Brunner}), in this section we outline the derivation of Bell's Theorem \cite{Bell} and formulate the assumptions of Free Choice and Local Causality in relation to two putative deterministic theories of physics: a conventional hidden-variable theory and IST. 

Two systems which have been produced by a common source are spatially separated and each measured by one of two distant experimenters: Alice and Bob. These experimenters can each choose one of two measurement settings, denoted by $X \in \{0,1\}$ and $Y \in \{0,1\}$ respectively. Once performed, the measurements yield outcomes $A \in \{0,1\}$ and $B\in \{0,1\}$ respectively. From one realisation of the experiment to another, the outcomes $A$ and $B$ vary, even for the same choices of measurement setting $X$ and $Y$. 

In both our conventional hidden-variable theory and IST we assume the existence of some supplementary variable $\lambda$ so that $A=A_{XY}(\lambda)$, $B=B_{YX}(\lambda)$ are deterministic formulae. Hence, with $\Lambda$ a finite sample space of supplementary variables (IST is an explicitly finite theory of quantum physics),
\begin{equation}
\label{prob}
E(AB|XY)= \sum_{\lambda \in \Lambda} A_{XY}(\lambda) B_{YX}(\lambda)\;  p(\lambda | XY) 
\end{equation}
denotes an expectation value for the product $AB$, and where $p(\lambda|XY)$ denotes a probability function on $\lambda \in \Lambda$. Based on these we can now define:
\begin{itemize}
\item \textbf{Free Choice}: $p(\lambda | XY)=p(\lambda)$
\item \textbf{Local Causality}: $A_{XY}(\lambda)=A_X(\lambda)$, $B_{YX}(\lambda) = B_Y(\lambda)$
\end{itemize}
Now consider the quantity (left-hand side of CHSH inequality):
\begin{equation}
\label{S}
S= E(AB|00)+E(AB|01)+E(AB|10)-E(AB|11)
\end{equation}
For our conventional hidden-variable theory which we assume satisfies Free Choice and Local Causality
\begin{equation}
S=\sum_{\lambda \in \Lambda} S(\lambda) \; p(\lambda)
\end{equation}
where
\begin{align}
S(\lambda)&=A_0(\lambda) B_0(\lambda)+A_0(\lambda) B_1(\lambda) +A_1(\lambda) B_0(\lambda)-A_1(\lambda) B_1(\lambda) \nonumber \\
&=A_0(\lambda)(B_0(\lambda)+B_1(\lambda))+A_1(\lambda)(B_0(\lambda)-B_1(\lambda)) \nonumber \\
&\le (B_0(\lambda)+B_1(\lambda)+|B_0(\lambda)-B_1(\lambda)| \le 2
\end{align}
Hence 
\begin{align}
\label{Sdet0}
S=\; &\sum_{\lambda \in \Lambda} A_0(\lambda) B_0(\lambda) p(\lambda)+
\sum_{\lambda \in \Lambda} A_0(\lambda) B_1(\lambda)p(\lambda)
\nonumber \\
+&\sum_{\lambda \in \Lambda} A_1(\lambda) B_0(\lambda)p(\lambda)-
\sum_{\lambda \in \Lambda} A_1(\lambda) B_1(\lambda)p(\lambda) \le 2
\end{align}
This is Bell's Theorem. 

Consider now a Bell experiment (e.g. \cite{Shalm}) where four separate sub-ensembles of particle pairs (each associated with some $\lambda$) are measured with the four selections  $XY=00, 01, 10, 11$ respectively. Let $\Lambda_{XY}$ denote the set of $\lambda$ values which occur for the specific experimental settings $X$, $Y$. Then, by Free Choice, our conventional hidden-variable theory predicts (since each $\Lambda_{XY}$ is individually a representation of $\Lambda$):
\begin{align}
\label{Sdet1}
S=\; &\sum_{\lambda \in \Lambda_{00}} A_0(\lambda) B_0(\lambda) p(\lambda)+
\sum_{\lambda \in \Lambda_{01}} A_0(\lambda) B_1(\lambda)p(\lambda)
\nonumber \\
+&\sum_{\lambda \in \Lambda_{10}} A_1(\lambda) B_0(\lambda)p(\lambda)-
\sum_{\lambda \in \Lambda_{11}} A_1(\lambda) B_1(\lambda)p(\lambda)
\end{align}
Experimentally $S$ is found to exceed $2$. Hence quantum physics cannot be described by a conventional hidden-variable theory which satisfies Free Choice and Local Causality. 

The second deterministic theory considered in this paper is IST (see Section \ref{IST}). To describe a key property of IST, we introduce the notation that $X^\dagger  $, $Y^\dagger $ denote the complementary values of $X$ and $Y$ respectively (e.g. if $X=0$, $X^\dagger  =1$ and so on). Then, for $\lambda \in \Lambda_{XY}$, IST has the emergent properties:
\begin{itemize}
\item If $A_{XY}(\lambda) \in \{0,1\}$ then $A_{X^\dagger  Y^\dagger }(\lambda) \in \{0,1\}$, but $A_{XY^\dagger }(\lambda)$ and  $A_{X^\dagger   Y}(\lambda)$  are undefined. 
\item If $B_{YX}(\lambda) \in \{0,1\}$ then $B_{Y^\dagger X^\dagger  }(\lambda) \in \{0,1\}$, but $B_{YX^\dagger  }(\lambda)$ and  $B_{Y^\dagger  X}(\lambda)$  are undefined. 
\label{ISTbullet1}
\end{itemize}
Consistent with this:
\begin{itemize}
\item $p(\lambda |XY)=p(\lambda | X^\dagger   Y^\dagger )=p(\lambda)$
\item $p(\lambda |X^\dagger  Y)=p(\lambda | X Y^\dagger )=0$
\label{ISTbullet2}
\end{itemize}
From the first set of bullet points, $A_{XY}(\lambda) \ne A_{X Y^\dagger }(\lambda)$. Hence Local Causality is manifestly false. Similarly, from the second set of bullet points $p(\lambda|XY) \ne p(\lambda| X Y^\dagger )$. Hence, Free Choice is also manifestly false. It is worth noting that the reasons for the failure of Free Choice and Local Causality are \emph{both} related to the undefinedness of $A$ and $B$ under certain conditions. This suggests something other than the failure of free choice and local causality may lie at the heart of the violation of the Bell inequality, since these are rather different physical concepts. 

In IST, 
\begin{align}
\label{Sdet2}
S = \;&\sum_{\lambda \in \Lambda_{00}} A_{00}(\lambda) B_{00}(\lambda)p(\lambda)
+\;\sum_{\lambda \in \Lambda_{01}}A_{01}(\lambda) B_{10}(\lambda)p(\lambda)
\nonumber \\
+&\sum_{\lambda \in \Lambda_{10}}A_{10}(\lambda) B_{01}(\lambda)p(\lambda)
-\sum_{\lambda \in \Lambda_{11}}A_{11}(\lambda) B_{11}(\lambda)p(\lambda)
\end{align}
Since IST satisfies neither Free Choice nor Local Causality, it is not constrained to satisfy $S \le 2$, consistent with the Bell Theorem. However, in Section \ref{free} we will weaken the definitions of Free Choice and Local Causality to `Free Choice on the Invariant Set' and `Local Causality on the Invariant Set'. What is meant by the `Invariant Set' is described in the next section. For now, it can be considered a subset of state space on which values of $A$ and $B$ are always well defined. As discussed in Section \ref{free} we will assert that despite being weaker,  `Free Choice on the Invariant Set' and `Local Causality on the Invariant Set' capture the essential space-time physics behind the notions of free choice and local causality. With these weakened definitions, the conventional hidden-variable theory continues to satisfy the Bell inequality as before, since by construction $A$ and $B$ are well defined on the whole of state space. Hence weakening Free Choice and Local Causality in this way does not change the conclusion that a conventional hidden-variable theory is inconsistent with experiment. However, this weakening does have a substantial impact on IST. For example, 'Local Causality on the Invariant Set' implies (see Section \ref{free}) that $A_{XY}(\lambda)=A_X(\lambda)$ and $B_{XY}(\lambda)=B_Y(\lambda)$ so that (\ref{Sdet2}) can be written
\begin{align}
\label{Sdet3}
S = \;&\sum_{\lambda \in \Lambda_{00}} A_{0}(\lambda) B_{0}(\lambda)p(\lambda)
+\;\sum_{\lambda \in \Lambda_{01}}A_{0}(\lambda) B_{1}(\lambda)p(\lambda)
\nonumber \\
+&\sum_{\lambda \in \Lambda_{10}}A_{1}(\lambda) B_{0}(\lambda)p(\lambda)
-\sum_{\lambda \in \Lambda_{11}}A_{1}(\lambda) B_{1}(\lambda)p(\lambda)
\end{align}
Now, (\ref{Sdet3}) is in an identical form to (\ref{Sdet1}). However, the key difference between (\ref{Sdet1}) and (\ref{Sdet3}) is that in the former case the four subsets $\Lambda_{XY}$ are statistically equivalent to one another (for any of the four correlations, the sum over $\Lambda_{XY}$ could be replaced by a sum over, say, $\Lambda_{X Y^\dagger }$). However, in IST, the four subsets $\Lambda_{XY}$ are not statistically equivalent (a sum over $\Lambda_{XY}$ certainly cannot be replaced by a sum over $\Lambda_{X Y^\dagger })$. Because of this we cannot write (\ref{Sdet3}) in the form $S= \sum_\lambda S(\lambda) p(\lambda)$ and so the Bell inequality cannot be established. As discussed below, it is claimed that the key reason for this difference arises from the fact that in a conventional hidden-variable theory, the set of well-defined and hence ontic states is the whole of (Euclidean) state space, whilst in IST the set of well-defined and hence ontic states is a nontrivial sub-set of state space. 
That is to say, a key conclusion of this paper is that it is state space geometry, rather than causality and free choice \emph{per se}, that lies at the heart of why quantum physics is so different from classical physics. 

\section{Invariant Set Theory}
\label{IST}
\begin{quote}
'The infinite is nowhere to be found in reality, no matter what experiences, observations, and knowledge are appealed to.' \cite{Hilbert}
\end{quote}

Following Hardy \cite{Hardy:2004} we can think of a typical problem in quantum physics as one where a physical system is prepared in some initial state, is then subject to various probability preserving unitary transformations, and is finally measured in the transformed state. In quantum theory, these transformations are, of course, described by the Schr\"{o}dinger equation. Motivated in part by Hilbert's observation above (reinforced more recently by \cite{Ellis:2018}), IST is an explicitly finite theory of quantum physics \cite{Palmer:2018a}. As in quantum theory, quantum systems in IST can be described by complex Hilbert vectors and tensor products. However, in addition these states must satisfy two finiteness criteria. In particular, the $n$-qubit Hilbert state,
\begin{equation}
\label{hilbert}
|\psi\rangle=\alpha_{00\ldots 0}|00 \ldots 0\rangle+\alpha_{00\ldots 1}|00 \ldots 1\rangle+\ldots +\alpha_{11\ldots 1}|11 \ldots 1\rangle
\end{equation}
written in either the preparation or measurement basis, is also a Hilbert state in IST providing
\begin{align}
\label{finite}
|\alpha_{ij\ldots k}|^2 &=\frac{n_1}{N}\nonumber \\
\frac{\arg{\alpha_{ij\ldots k}}}{2\pi}&=\frac{n_2}{N}
\end{align}
where $i,j,\ldots, k \in \{0,1\}$. Here $N\gg1$ is some large but finite whole number. As suggest in \cite{Palmer:2018a}, the largeness of $N$ may be related to the relative weakness of gravity, though for the purposes of this paper, $N$ is simply treated as an arbitrarily large integer, and $n_1, n_2$ are integers such that $0\le n_1, n_2 \le N$.  Importantly (see below), $N$ is larger than the maximum number $M_1$ of bits with which an experimenter can determine measurement orientations using finite precision measuring apparatuses, or the maximal size $M_2$ of any computational machine potentially available to the experimenters. Below we refer to Hilbert states which satisfy the finiteness conditions (\ref{finite}) as `finite Hilbert states'. A simple example is the single qubit
\be
\label{qubit}
\cos\frac{\theta}{2} |0\rangle + \sin\frac{\theta}{2} e^{i \phi} |1\rangle
\ee
where the finiteness conditions (\ref{finite}) demand that both $\cos \theta$ and $\phi/2\pi$ are rational numbers.  Using Niven's theorem concerning number-theoretic properties of the cosine function (see below) this provides the basis of quantum complementarity (and as discussed below, Bell `nonlocality') in IST \cite{Palmer:2018a}. In IST the transformations between preparation and measurement are effected by a finite approximation to the Schr\"{o}dinger equation (i.e. based on the finite calculus) which guarantees the transformed final states satisfy the finiteness criteria (\ref{finite}). 

The set of Hilbert states satisfying (\ref{finite}) is manifestly not closed under addition (see Appendix B). However, in IST, algebraic closure is reinstated at deeper deterministic level (by a simple argument described in Appendix A, any finite theory of physics cannot be indeterministic). As discussed in \cite{Palmer:2018a}, Hilbert states satisfying (\ref{finite}) are statistical representations of a fractal geometry $I_U$ of state-space trajectories (or `histories') in cosmological state space. In particular, for each degree of freedom, a trajectory at one fractal iterate comprises a helix of $N$ trajectories at the next highest fractal iterate (see Fig \ref{helix} c.f. the way rope is constructed). This fractal geometry $I_U$ is referred to as the `invariant set' because if a state lies on $I_U$ at some particular time, under dynamical evolution it will always lie on it. Conversely if a state does not lie on $I_U$ at a particular time, then under dynamical evolution it will never lie on $I_U$. In this sense $I_U$ can be compared with the fractal attractors of classical chaotic systems. However, unlike such classical systems, $I_U$ is considered a primal geometric entity from which the dynamical laws of evolution derive. In particular, unlike classical theory, the ontic states of IST are the states which lie on $I_U$ and only them. The details of $I_U$ are unimportant for the purposes of this paper. All that is relevant is this:
\begin{itemize}
\item IST is a finite (and hence deterministic) theory, in which state space is partitioned into a subset $I_U$ of ontic states and a complement of undefined non-ontic states. 
\item If and only if a Hilbert State (\ref{hilbert}) satisfies the finiteness condition (\ref{finite}) does it have a statistical, i.e. probabilistic, representation on a finite sample space of ontic states of IST. 
\end{itemize}
 
\begin{figure}
\centering
\includegraphics[scale=0.2]{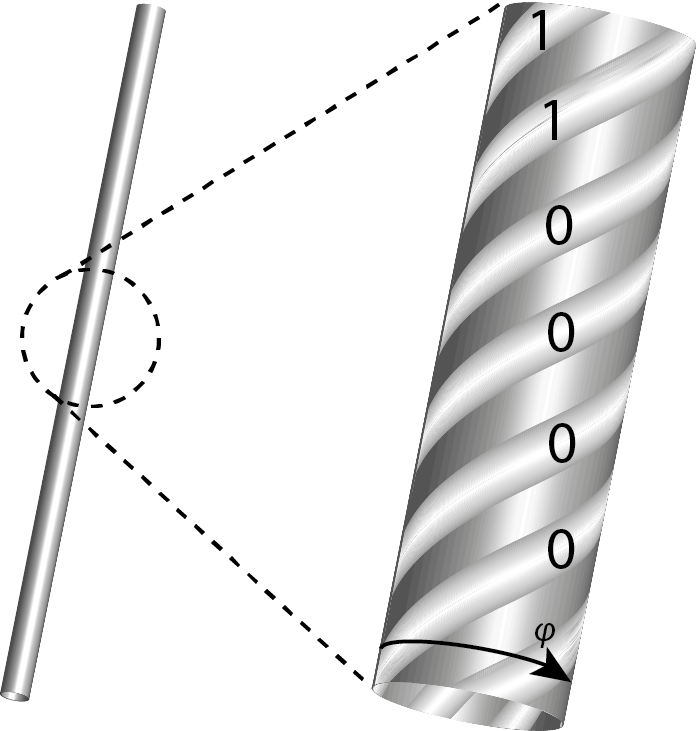}
\caption{\emph{A state-space trajectory segment, which appears to be a simple line on some coarse scale, is in fact found to be, on magnification, a helix of $N$ trajectory segments. On further magnification (like a section of rope), each of these helical trajectory segments is itself a helix of $N$ trajectories, and so on. A cross section through the original coarse-scale trajectory segment has a Cantor set structure based on $N$ nested disks. At any particular level of magnification (i.e. fractal iterate) the trajectory segments can be labelled $0$ or $1$ according to the distinct regions of state space to which they evolve under decoherence, i.e. under interaction with the environment, and possibly cluster under the attractive effects of gravity \cite{Penrose:2004} \cite{Diosi:1989}. A statistical representation of this fractal helix is given by the finite Hilbert vector (\ref{qubit}) where $\phi$ denotes a rotation of the helix and $\cos^2 \theta/2$ determines the fraction of helical trajectory segments labelled $0$.}}
\label{helix}
\end{figure}

In IST, finite Hilbert states always have a probabilistic interpretation in terms of the geometry of $I_U$ (Appendix B shows that superpositions of finite Hilbert states are not themselves finite and hence not ontic). In particular the /preparation/measurement eigenstates $|00\ldots0\rangle$, $|00\ldots1\rangle$ etc refer to distinct regions of state space to which individual state space trajectories evolve (and indeed cluster) through the process of decoherence, i.e. environmental coupling. As such, individual trajectories can be labelled by the cluster to which they evolve. Like the riddled basins of attraction in nonlinear dynamical system theory \cite{Ott}, neighbouring trajectories of the helix can be labelled according to the cluster into which they evolve. Because the finite Hilbert states always have a probabilistic interpretation, there is no state-vector collapse and hence no `measurement problem' in IST. See \cite{Palmer:2018a} for more details. 

\section{Bell's Theorem and Invariant Set Theory}
\label{bellist}
We start this section by reviewing the quantum theoretic description of a Bell experiment. In particular, each particle pair is in the spherically symmetric singlet state  
\begin{equation}
|\psi\rangle=\frac{1}{\sqrt 2}(|\hat{\bm{n}}, 0 \rangle_1 |\hat{\bm{n}}, 1 \rangle_2 - |\hat{\bm{n}}, 1 \rangle_1 |\hat{\bm{n}}, 0 \rangle_2 )
\end{equation}
for unit vectors $\hat{\bm{n}}$ in physical 3-space. Suppose spin measurements are performed on particles 1 and 2 in the directions $\hat{\mathbf{a}}$, $\hat{\mathbf{b}}$, respectively, so that in the measurement basis, the expectation value 
\begin{align}
\label{pauli}
E(\hat{\mathbf{a}},\hat{\mathbf{b}})& =\langle \psi | (\bm{\sigma}.\hat{\mathbf{a}})(\bm \sigma.\hat{\mathbf{b}}) | \psi \rangle
\end{align} 
where $\bm \sigma$ are Pauli matrices. 

Now the settings $X=0,1$ and $Y=0,1$ correspond to measurement orientations and hence to points on the unit celestial sphere. Consider in particular the three points $X=0$, $X=1$ and $Y=0$ represented by the vertices of a spherical triangle $\triangle$ in Fig \ref{F:CHSH}. Let $\hat {\mathbf x}$, $\hat {\mathbf y}$ and $\hat {\mathbf z}$ represent unit vectors relative to Cartesian coordinates $Oxyz$.  Let us orient the coordinates so that $X=0$ corresponds to the unit vector $\hat {\mathbf a}_{X=0}=\hat {\mathbf z}$ and $X=1$ is the unit vector $ \hat{\mathbf a}_{X=1}=\sin \theta_{X=0 X=1} \hat {\mathbf y}+\cos \theta_{X=0 X=1} \hat {\mathbf z}$, where $\theta_{X=0 X=1}$ denotes the angular distance on the celestial sphere between $X=0$ and $X=1$ (with obvious generalisation below). Using this, $Y=0$ corresponds to the unit vector
$$
\hat{\mathbf b}_{Y=0}=\sin\theta_{X=0Y=0} \sin \gamma\; \hat {\mathbf x}+\sin\theta_{X=0Y=0} \cos \gamma \; \hat {\mathbf y}+\cos\theta_{X=0Y=0} \; \hat {\mathbf z}
$$
where $\gamma$ is the internal angle of the spherical triangle $\triangle$ at the vertex $X=0$ (see Fig \ref{F:CHSH}). Based on this, the three spin operators for measurements along directions corresponding to the vertices of the triangle are:
\begin{align}
\bm{\sigma}.\hat{\mathbf{a}}_{X=0}=&
\begin{pmatrix}
1 & 0 \\
0 & -1
\end{pmatrix}
\nonumber \\
\bm{\sigma}.\hat{\mathbf{a}}_{X=1}=&
\begin{pmatrix}
\cos \theta_{X=0 X=1} & - \sin \theta_{X=0 X=1} \\
\sin \theta_{X=0 X=1} & \cos \theta_{X=0 X=1}
\end{pmatrix}
\nonumber \\
\bm{\sigma}.\hat{\mathbf{b}}_{Y=0}=&
\begin{pmatrix}
\cos \theta_{X=0 Y=0} & - e^{-i \gamma} \sin \theta_{X=0 X=1} \\
e^{i \gamma} \sin \theta_{X=0 Y=0}  & \cos \theta_{X=0 Y=0}
\end{pmatrix}
\end{align}
with eigenvectors
\begin{align}
\begin{pmatrix}
1  \\
0 
\end{pmatrix}&\  \text{and}\ 
\begin{pmatrix}
0  \\
1 
\end{pmatrix}
\nonumber \\
\begin{pmatrix}
\cos \frac{\theta_{X=0 X=1}}{2} \\
\sin \frac{\theta_{X=0 X=1}}{2}
\end{pmatrix}& \  \text{and}\  
\begin{pmatrix}
\sin \frac{\theta_{X=0 X=1}}{2} \\
\cos \frac{\theta_{X=0 X=1}}{2}
\end{pmatrix}
\nonumber \\
\begin{pmatrix}
e^{-i \gamma} \cos \frac{\theta_{X=0 Y=0}}{2} \\
\sin \frac{\theta_{X=0 Y=0}}{2}
\end{pmatrix}& \  \text{and}\  
\begin{pmatrix}
e^{-i \gamma} \sin \frac{\theta_{X=0 Y=0}}{2}  \\
\cos \frac{\theta_{X=0 Y=0}}{2}
\end{pmatrix}
\end{align}
Now in quantum theory, the singlet state, transformed to the measurement basis, varies continuously as $\theta_{XY}$ and $\gamma$ range over the reals. This property of continuity is fundamental to quantum theory, as highlighted in Hardy's axiomatic approach to quantum theory \cite{Hardy:2004}. However, this property breaks down in an important way in IST. 

\begin{figure}
\centering
\includegraphics[scale=0.3]{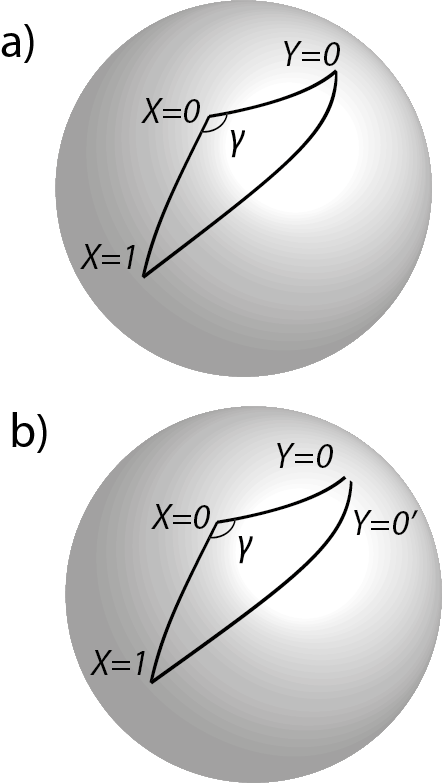}
\caption{\emph{a) In general it is impossible for all the cosines of the angular lengths of all three sides of the spherical triangle $\triangle$ to be rational, and the internal angles to be rational multiples of $2\pi$. That is, the finiteness condition (\ref{finite}) cannot be satisfied for a counterfactual measurement $X=1$, $Y=0$, when it is satisfied for a realisable measurement $X=0$, $Y=0$ on $I_U$. Hence $p(\lambda | 10)=0$ when $p(\lambda | 00) \ne 0$ which implies invariant set theory must violate Free Choice, consistent with the fact that the theory does not satisfy the Bell inequality.  b) When a measurement $X=1$, $Y=0$ is made on $I_U$ (on a different particle pair), the cosine of the angular length between $X=1$ and $Y=0$ must now be rational. Hence, the precise orientation associated with the first $Y=0$ measurement cannot not be the same as the precise orientation for the second $Y=0$ measurement (represented in the figure as $Y'=0$). In a precise sense, a) is a singular limit of b) as the angle between $Y=0$ and $Y'=0$ is set to zero \cite{Palmer:2018c}. Since gravitational waves represent an irreducible source of noise, these considerations suggest a deep synergy between IST and gravitation theory.}}
\label{F:CHSH}
\end{figure}

To see this, note that the finiteness conditions (\ref{finite}) applied to the (measurement basis) eigenvectors imply that $\cos \theta_{X=0 Y=0}$, $\cos \theta_{X=0 X=1}$ and $\gamma/2 \pi$ must be rational numbers. This is a requirement for the Hilbert state to represent, in probabilistic form, the outcome of two types of experiment on $I_U$ to be well defined: a) Alice measured her particle with the $X=0$ setting and Bob measured the spin of his particle with the $Y=0$ setting and b) Alice measured the spin of her particle with the $X=0$ setting and with the spin prepared in this way, performed a second spin measurement by sending the same particle through the measuring apparatus with $X=1$ selected. 

However, suppose we ask the counterfactual question: given the possibility of measurements a) and b), is it possible that Alice and Bob could have measured the spin of their particles with respect to $X=1$ and $Y=0$ respectively. In quantum theory the answer is yes and (by Hardy's Continuity Axiom \cite{Hardy:2004}) quantum theory provides a probabilistic answer. Here we can transform the Cartesian axes so that the unit vector $\hat{\mathbf z}$ coincides with $X=1$ and the direction associated with $Y=0$ lies on the $(\hat{\mathbf y}, \hat{\mathbf z})$ plane. For this to have an ontic representation in IST, then this implies that $\cos \theta_{X=1 Y=0}$ must also be rational. Is it possible for $\cos \theta_{X=1 Y=0}$ to be rational, given that $\cos \theta_{X=0 Y=0}$, $\cos \theta_{X=0 X=1}$ and $\gamma/2\pi$ are rational? 

We now introduce Niven's theorem, essential for the analysis that follows. A simple proof is given in Appendix C. 
\begin{theorem}
\label{theorem2}
 Let $\phi/\pi \in \mathbb{Q}$. Then $\cos \phi \notin \mathbb{Q}$ except when $\cos \phi =0, \pm \frac{1}{2}, \pm 1$. \cite{Niven, Jahnel:2005}
\end{theorem}

Using Niven's theorem we can decide on the rationality of $\cos \theta_{X=1 Y=0}$ by using the cosine rule for spherical triangles,
\be 
\label{cosinerule}
\cos \theta_{X=1 Y=0}=\cos \theta_{X=0Y=0} \cos \theta_{X=0X=1} + \sin \theta_{X=0Y=0} \sin \theta_{X=0X=1} \cos \gamma
\ee
applied to $\triangle$. Now by Niven's theorem $\cos \theta_{X=1 Y=0}$ would be rational if the three points $X=0,1$ and $Y=0$ lay on a great circle \emph{exactly}, so that $\gamma=180^\circ$ precisely. Now we can certainly assume that, by experimental design, the three points \emph{nominally} lie on a single great circle corresponding to coplanar orientations of the measuring devices. However, given a sample space of $K$ equally likely angles in a small neighbourhood of $180^\circ$, the likelihood of choosing $180^\circ$ precisely scales as $1/K$ and since $K$ scales as $N$ then the likelihood becomes arbitrarily small for large enough $N$. As soon as $\gamma$ deviates from $\pi$ by the smallest amount conceivable, its cosine will not be rational. Moreover, since Alice and Bob's measuring apparatus' have finite precision they will be unable to contrive for the three orientations to be coplanar \emph{precisely}. Hence, in general, $\cos \theta_{X=1 Y=0}$ is the sum of two terms, the first a rational and the second the product of three independent terms, the last of which is irrational. Being independent, these three terms cannot conspire to make their product rational. Hence $\cos \theta_{X=1 Y=0}$ is the sum of a rational and an irrational and must therefore be irrational. 

More generally, this result shows that if Alice and Bob make the choices $X$ and $Y$ for a given particle pair characterised by $\lambda$, then the counterfactuals where they instead made the choices $X$,  $Y^\dagger $ or $X^\dagger  $, $Y$ are undefined for that particular particle pair. As such, the states associated with the latter counterfactuals cannot be described by complex Hilbert states that satisfy the finiteness condition (\ref{finite}) and therefore do not correspond to ontic states on the invariant set $I_U$. (There is nothing in the argument above, however, to say that Alice and Bob could not have chosen $X^\dagger  $, $Y^\dagger $, having actually chosen $X$ and $Y$). The metaphysical implications of these conclusions are deferred to Section \ref{free}. The crucial role of counterfactuals in the Bell theorem has been highlighted elsewhere \cite{tHooft:2015b}.

The analysis of this section justifies the IST bullet summary in Section \ref{bell} and hence the fact that IST satisfies neither Free Choice nor Local Causality. However, demonstrating this failure has required us to consider explicitly states which are not ontic. Similar considerations apply to Local Causality. If $A_{00}(\lambda)$ has a definite value, either 0 or 1, then by construction $A_{01}(\lambda)$ is undefined. Hence $A$ cannot be of the form $A_X(\lambda)$. Identical arguments apply to $B$. From a physical point of view this requirement to consider non-ontic states explicitly is unsatisfactory - the notions of free choice and local causality to refer to properties of the real world and should not be dependent on putative unphysical worlds. This suggests a reformulation Free Choice and Local Causality is needed, so that they only refer explicitly to ontic states. This is done in the next Section. 

To conclude this section, suppose that after the first pair of measurements $X=0$, $Y=0$ on one particle pair, a pair of measurements $X=1$, $Y=0$ was made on a second particle pair. How could this be possible if $\cos \theta_{X=1 Y=0}$ is irrational? The answer is clearly that the precise orientation associated with $Y=0$ for the first particle-pair measurement need not be identical with the precise orientation $Y=0$ associated with the second particle-pair measurement. Of course there is no reason why these orientations should be identical - the associated measurements would have been performed at different times and/or spatial locations. At the very least, the generic presence of uncontrollable gravitational waves would ensure the the two precise orientations associated with $Y=0$ for the two particle pairs will not be the same. That is to say, the presence of such irreducible background `noise' is consistent with the requirement in IST that for any settings $X$, $Y$ measuring an entangled particle pair, $\cos \theta_{XY}$ must satisfy the finiteness condition (\ref{finite}) and therefore be rational. (By the fact that Alice and Bob's measurements have finite precision, it is impossible for them to contrive to somehow violate the rationality of $\cos \theta_{XY}$ in practice.) In Fig \ref{F:CHSH}b) we label these two realisations of $Y=0$, with the second point labelled $Y'=0$. Although, for large enough $N$, $Y'=0$ can become arbitrarily close to $Y=0$ (in the Euclidean sense), the finiteness condition (\ref{finite}) ensures that the points can never be precisely the same. 

In an exact sense, $Y=0$ is the singular limit \cite{Berry} of $Y'=0$ as noise from gravitational waves and other effects are set precisely to zero \cite{Palmer:2018c} - suggesting a possible deep synergy between IST and gravitation theory (see Section \ref{discussion}). As such, we can look at the violation of the Bell inequality in a different way \cite{Palmer:2018c}. If we think of the Bell inequality as defined by the form (\ref{Sdet0}) that is satisfied by a conventional hidden-variable theory, then in IST the Bell inequality is neither satisfied nor violated, it is  undefined. The form (\ref{Sdet3}) that is violated is (in the $p$-adic sense, see below) distant from the form (\ref{Sdet0}). 

\section {Weakened Definitions of Free Choice and Local Causality} 
\label{free}
\subsection{Free Choice and Local Causality on the Invariant Set}
In this Section we modify the definitions of Free Choice and Local Causality so that they apply only to ontic states of our deterministic theories. Specifically:
\begin{itemize}
\item \textbf{Free Choice on the Invariant Set}: $p(\lambda | XY)=p(\lambda)$ for triples $(X, Y, \lambda)$ corresponding to states on $I_U$.
\item \textbf{Local Causality on the Invariant Set}: $A_{XY} (\lambda)=A_X(\lambda)$, $B=B_Y(\lambda)$ for triples $(X, Y, \lambda)$ corresponding to states on $I_U$.
\end{itemize}

Our conventional hidden-variable theory must still satisfy the Bell inequality even with these weaker definitions of Free Choice and Local Causality. The reason is straightforward. In such (classical) hidden-variable theories, there is no special ontological subset of state space; $I_U$ can be identified with the whole of state-space. From this perspective, all counterfactual states lie on this trivial invariant set. That is to say, the restriction of states to those lying on $I_U$ is no restriction at all and there is no difference between `Free Choice' and `Free Choice on the Invariant Set'. Similarly, there is no difference between `Local Causality' and `Local Causality on the Invariant Set'. 

However, such a restriction has significant implications in IST. According to IST, if $(\lambda, X, Y)$ is a triple corresponding to an (ontic) state on $I_U$, then so is $(\lambda, X^\dagger  , Y^\dagger )$, but neither $(\lambda, X^\dagger  , Y)$ nor $(\lambda, X, Y^\dagger )$ is a state on $I_U$. Hence, Free Choice on the Invariant Set is equivalent to the statement that $p(\lambda | XY)=p(\lambda| X^\dagger   Y^\dagger )$ for all choices $X$, $Y$ on $I_U$ which is true. The problematic situations $p(\lambda|X Y^\dagger )$ and $p(\lambda| X^\dagger   Y)$ which are equal to zero when $p(\lambda|XY)$ is nonzero, do not arise on $I_U$ by construction. In IST, $p(\lambda| XY)=p(\lambda | X^\dagger   Y^\dagger )=p(\lambda)$ on $I_U$ hence IST satisfies `Free Choice on the Invariant Set'. For similar reasons, IST satisfies `Local Causality on the Invariant Set'- the notion of an undefined value for $A$ or $B$ never occurs on $I_U$ by construction. Hence, the situations where $A_{XY}(\lambda) \ne A_{X Y^\dagger}(\lambda)$ never occur on $I_U$. That is to say, on $I_U$, $A$ is determined from $X$ and $\lambda$ only. 

Despite this, it might be thought that the finiteness condition (\ref{finite}) puts a constraint on the angles Alice could choose (for a given orientation of Bob's measuring device), contradicting the notion that she is a free agent. However, both measuring apparatus' have finite precision, denoted by the parameter $M_1$ (the larger is $M$ the more precise the orientations can be set). If $N=M_1+\Delta_1$ ($\Delta_1>0$), there is \emph{nothing} Alice can actually do to violate the finiteness conditions on $I_U$ (\ref{finite}), no matter how large is $M_1$ (i.e. no matter how precise are the measuring apparatuses). In this sense, for all practical purposes, there are no orientations that are not available to Alice when she orients her measuring device. For all practical purposes she is indeed a free agent. 

In a similar sense, one could perhaps imagine trying to subvert Free Choice on the Invariant Set by computing (in the real world) whether or not a putative value $A_{XY}(\lambda)$ has a defined value or not and choosing the settings $X$ and $Y$ where $A$ is undefined. If this was a computational problem, IST would be inconsistent with Free Choice on the Invariant Set. However, it is not a computational problem if $N=M_2+\Delta_2$ where now $M_2$ represents the power of some maximal (but necessarily finite) computational system on $I_U$. This means that estimating whether a putative state lies on $I_U$ or not is effectively uncomputable (computationally irreducible \cite{Wolfram} to be precise, see also \cite{Blum} \cite{Dube:1993}). That is, there is no computation on $I_U$ which can output the value $1$ when $A_{XY}(\lambda)$ is associated with a state on $I_U$ and output the value $0$ when $A_{XY}(\lambda)$ is not associated with a state on $I_U$ and is therefore undefined. Because of this, IST is not computationally inconsistent with Local Causality on the Invariant Set. 

\subsection{Are These Weakened Definitions Physically Reasonable?}

Are `Free Choice on the Invariant Set' and `Local Causality on the Invariant Set' reasonable definitions of free choice and local causality as we understand these matters from a physical perspective? Consider the notion of human free will more generally. This is frequently described as an ability to have done otherwise, something we feel rather viscerally to be true in many situations. By assuming we could have done otherwise, we admit a theoretical framework where certain degrees of freedom can be perturbed, keeping all others fixed. For example, by saying I could have turned left at the junction, having actually turned right, I am implicitly assuming a world $U'$ identical in all respects to the real one $W$ except that I turned left, such that $W'$ is consistent with the laws of physics just as $W$ is. However, why should $W'$ be consistent? For example, is there necessarily a cosmic initial state which, given the laws of physics, will produce such a $W'$? And if $W'$ is not consistent, should we therefore conclude that we do not have free will? For many people this would be a ludicrous conclusion to draw, because for them the notion of free will merely describes an absence of constraints preventing them from doing what they want to do, e.g. turning right. Of course there are situations where they may not be free to do what they want (paying taxes, dying prematurely) but that does not mean there aren't many situations where they are free to do what they want. If we define free will in this latter sense, then we can talk meaningfully about free will without invoking the notion of counterfactual definiteness and hence we can talk meaningfully about free will when counterfactual definiteness is only partially true, as here. In this sense, the counterfactual definition of free will is stronger than the one based on absence of constraints; the former implies the latter, but not the converse. In this sense, the weakened definition of Free Choice on the Invariant Set does not restrict the notion of free choice meant as an absence of constraints preventing Alice and Bob from performing the measurements they want to perform (i.e. defined as a constraint on events in space-time). In particular, in IST, if Alice wants to choose $X=0$ she can so do. If she wants to choose $X=1$, she can so do. Similarly for Bob. As such, Free Choice on the Invariant Set is a physically reasonable weakening of the overly-strong definition of Free Choice. 

The notion of causality can be analysed similarly. If we clap our hands and hear the echo from a nearby wall a fraction of a second later, we can say that the echo was caused by the clap in the sense that we believe that if we had not clapped on an occasion when we did clap, no sound would have been heard. But again, if the laws of physics do not always permit such counterfactual worlds, does this invalidate the notion of causality? Certainly not, because we can instead solve the relevant equations (in space-time) to infer that the clap is the source of a propagating acoustic oscillation, reflected from the nearby wall and transmitted to our ear where it excites a neuronal response in our brains. It is the latter idea, not the former, that is the principal determinant of the notion of causality in physics. Indeed in relativity theory causality is built into the very metric fabric of space-time - it has nothing to do with the definiteness or otherwise of counterfactuals in state space. Again, defining causality which (implicitly or otherwise as in the case of IST) depends on counterfactual situations, is too strong a constraint when developing models of quantum physics which are consistent with relativistic space-time. 

\section{Possible Objections to Invariant Set Theory}
\label{objections}

\subsection{Conspiracy?}

Putative deterministic accounts of the violation of the Bell Inequality are generally seen as implausibly conspiratorial (e.g. \cite{WisemanCavalcanti}). A classic example illustrating such a possibility is where the values for $X$ and $Y$ are determined by the wavelength of light from distant widely separated quasars. It seems utterly implausible to imagine that there could be some correlation between the hidden variables associated with some entangled particle pair emitted by a local laboratory source, and the wavelength of light from such astronomical sources. 

However, this supposes that the hidden variables are themselves localisable (e.g. within the particles in question) and ignores the fact that some decision was made to use light from two quasars as the determinants of $X$ and $Y$ respectively, rather than, say, bits from the works of Shakespeare and Goethe (or a million other whimsical determinants). From where does the information come that determines the decision to use quasar light?

Let us explore these two issues in turn. Consider the following analogy. When a baby is born, its sex can be classified as male or female, and can be determined from information internal to the baby (e.g. its DNA). By contrast, the gender of the baby is not determined in such a biological way. Babies can be said to belong to one of two sets: those who by the age of (say) 40 will have identified as having either male or female gender.  In a deterministic (relativistic) theory of the universe, the information needed to determine to which of these sets a baby belongs may not exist internally, but nevertheless will exist on the intersection of a spacelike hypersurface containing the birth event, with the past light cone of the 40th birthday event. In this sense, the information which determines to which gender set the baby belongs may not be localisable within the baby, but will still exist within the relevant light cone. 

In this paper, the relevant space-time event is not one of gender identification, but a quantum spin measurement. Nevertheless, the same basic conclusion applies: in a deterministic and relativistic world, we should certainly not assume that the supplementary information contained in $\lambda$ is somehow localised within the particle; it will instead by delocalised in the intersection $\mathcal S$ of a spacelike hypersurface containing the particle's moment of creation, with the past light cone of the measurement event. Note, incidentally, that nowhere in the discussion so far have we specified $\lambda$ other than it is supplemental to discrete measurement settings in determining some measurement outcome. That leaves $\lambda$ open to quite broad interpretation. 

Returning to the second point, one can ask where on $\mathcal S$ is located the information which determines that the experimenters will choose the wavelength of quasars rather than bits from the works of Shakespeare or Goethe. Like $\lambda$, the information is not localisable on $\mathcal S$ but is spread across $S$. Indeed, following ideas in nonlinear systems theory \cite{Ott}, for example on riddled basins of attraction, one can expect that the two types of information on $\mathcal S$ ($\lambda$ on the one hand, choosing between quasars or books on the other) are intertwined in the sense that a small perturbation in any localised region of $\mathcal S$ will impact both $\lambda$ and the decision on which source of information to use. Of course the wavelength of light and the bits of Shakespeare are themselves localised information, but everything else that determines whether it is to be wavelength of light, bits of Shakespeare (or a million other whimsical factors) is delocalised. 

A manifestation of the profoundly intertwined nature of these two sources of information is made explicit when considering the counterfactual question: what would the experimenter have observed had she chosen bits of Shakespeare on an occasion when she in fact chose the wavelength of quasar light. This is precisely the type of counterfactual that is disallowed by IST (when these choices determine quantum measurement settings). In particular, a perturbation to the world which changes the decision from quasar bits to Shakespeare bits, keeping the particles to be measured unchanged, will take a state of the world off $I_U$ to a non-ontic state. As discussed in Section \ref{free}, a denial of counterfactual definiteness in this way does not invalidate the notion of free will in the least, since the decision to use the wavelength of photons was made without the experimenter feeling any constraint to do other than what they actually did. 

These issues do, nevertheless, indicate a holistic nature to the world, similar to when it was concluded that respecting the finiteness condition (\ref{finite}) in the presence of gravitational waves suggests some synergy between quantum physics and gravitation theory. However, holism should not be seen as tantamount to either implausible conspiracy or nonlocality (in the sense of a negation of local causality). Rather, these matters are consistent with this notion that ontic states lie on and hence are constrained to a nontrivial fractal geometry in state space. 

\subsection{Fine Tuning?}
\label{fine}

Using the Euclidean metric of Hilbert Space, two Hilbert states, one satisfying (\ref{finite}) and one not, can be arbitrarily close for large enough $N$. If we only associate ontic states with finite states, then a theory based on finite Hilbert states (as defined in (\ref{finite})) would appear to be fine tuned \cite{WoodSpekkens}. This appears unacceptable as we require theories of physics to be robust to small amounts of noise (which should therefore map ontic states to ontic states). However, as discussed, there is a deterministic underpinning to IST based on the primacy of a  fractal state-space geometry $I_U$. The primacy of fractal geometry in particular suggests that the Euclidean metric might not be the physically appropriate way to measure distances in state space. Since $p$-adic numbers are to fractals as real numbers are to Euclidean geometry (the set of $p$-adic integers is homeomorphic to Cantor sets with $p$ iterated pieces \cite{Katok}), then the $p$-adic metric (with $p=N$) is more respectful of the primacy of $I_U$ than is the traditional Euclidean metric (e.g. of Hilbert Space). In particular, non-ontic states which do not lie on $I_U$ cannot in the $p$-adic metric be close to ontic states which do lie on $I_U$ (even though from a Euclidean perspective such states may appear arbitrarily close). This is a conceptually important conclusion for the metaphysics of counterfactuality. In terms of the corresponding $p$-adic norm, it is impossible for small-amplitude noise to perturb an ontic state on $I_U$ to a non-ontic state off $I_U$. Hence, from the more physical $p$-adic perspective, IST is not fine tuned at all.  The metrics of classical and quantum theory (whose state spaces are both continuum spaces) arise at the singular limit $p=\infty$ of the $p$-adic metric. These matters are discussed in more detail in \cite{Palmer:2018a}. 

\section{Discussion}
\label{discussion}

Bell's theorem is the epitome of why quantum physics is perceived as being incomprehensible (a notion which has percolated into public consciousness). The theorem appears to imply that any theory in which two seemingly reasonable assumptions hold - one about the freedom of experimenters to choose measurement settings without having their minds corrupted by the value of particle hidden variables, and the other that outcomes cannot depend on events which occur arbitrarily far away and hence outside the light-cone - must necessarily be inconsistent with experiment. However, of course, the theorem actually depends critically on the precise mathematical definitions of Free Choice and Local Causality as they are known in the literature. A principal conclusion of this paper is that these definitions are too strong and lead to theories which would obey free choice and local causality in any practical sense of the word being rejected. That is to say, it may be that in our search to understand the Bell Theorem, we are throwing the baby out with the bathwater. 

In this paper we have described a finite model of quantum physics (invariant set theory - IST - \cite{Palmer:2018a} and have shown that certain counterfactual states necessarily lie in gaps in a fractal state-space geometry $I_U$, assumed primal in IST's theoretical framework. Precisely because they lie in the gaps, such types of counterfactual states cannot be ontic. As discussed in \cite{Palmer:2018a}, for quantum physical experiments (like quantum interferometry, the sequential Stern-Gerlach experiment and GHZ) these specific types of non-ontic counterfactual states are generic in quantum physics. Here we have shown that they also arise in the interpretation of the Bell Theorem. In particular, they necessarily arise in the definitions of Free Choice and Local Causality. It is concluded that these definitions are indeed too strong and rule out theories which have free choice and local causality in any practical operational sense - a notion which is defined precisely as `Free Choice on the Invariant Set' and `Local Causality on the Invariant Set'. As such, we claim that the Bell Theorem does not rule out deterministic theories which exhibit free choice and local causality in any physically meaningful sense. For very similar reasons we conclude that IST is not superdeterminstic in the sense of relying on implausible conspiratorial correlations. Neither is it fine-tuned (because it is underpinned by a $p$-adic metric which necessarily respects the primacy of a fractal state-space geometry, implying that non-ontic states can never be close to ontic states). 

The principal motivation for this paper has not been to make sense of the Bell Theorem \emph{per se}. Rather, it has been to formulate a theory of quantum physics which is more obviously compatible than is quantum theory with the deterministic nonlinear causal theory of general relativity. The failure to synthesise quantum and gravitational physics convincingly has been one of the great failings of theoretical physics over the last half century. The results in this paper suggest that with a suitably formulated alternative to quantum theory, there is no longer any incompatibility with relativistic causality, in spirit even if not in practice \cite{Penrose:1989}.  As such, the required synthesis, instead of requiring any fundamental revision of our concepts of causality, may require an extension of the primality of geometry, the basis of general relativity, from space-time to (cosmological) state space. The type of holistic primal geometry $I_U$ discussed in this paper provides an example of the notion of a `top-down' influence \cite{Ellis} where some geometric constraint on the state space of the whole universe determines the properties of individual particles. This is in complete contrast to the more usual bottom-up philosophy of reductionism. The author will continue to pursue a programme of research applying such holistic top-down concepts to some of the problems at the cutting edge of fundamental physics.

\section*{Acknowledgements}
The author is very grateful for discussions with Roger Colbeck, Oliver Reardon-Smith and Tony Sudbery following a seminar at the University of York. These led the author to realise that one needs to be quite explicit about how definitions of free choice and local causality differ from those conventionally used, if one is  claiming a locally realistic theory which violates the Bell inequality.  It was these discussions that provided the motivation to write this paper. The author is also grateful to Felix Tennie for helpful comments on an early draft of the paper. 

\section*{Appendix A}
A finite theory generates bit strings $S$ whose maximal information content is finite. 
\begin{definition}
Any putative finite theory of physics is irreducibly indeterministic if such a maximal $S$ cannot be generated deterministically. 
\end{definition}
\begin{theorem}
A putative finite theory of physics cannot be irreducibly indeterministic. 
\end{theorem}
\begin{proof}
Let 
\begin{equation} 
S=(a_1, a_2 \ldots a_N)
\end{equation}
denote a maximal bit string where the brackets denote the possibility of periodicity (like the digits of the rational number 1/7). By construction, this specific bit string \emph{can} be generated by the deterministic procedure
$$
a_n=\lfloor 2r_{n-1} \rfloor
$$
where
$$
r_n=2r_{n-1}-\lfloor 2 r_{n-1} \rfloor
$$
where $0 < n \le N$. Here $\lfloor \ldots \rfloor$ denotes the integer floor function, and $r_0$ is the rational number whose base-2 expansion is $0.a_1a_2a_3 \ldots a_N$.\end{proof}
\section*{Appendix B}
Here we show that in IST the normalised sum of two finite Hilbert States is generically not finite. Let us add the two (normalised) finite Hilbert States
\begin{align}
|\psi_1\rangle=\frac{1}{\sqrt 2} (|0\rangle + e^{i \phi_1} |1\rangle \nonumber \\
|\psi_2\rangle=\frac{1}{\sqrt 2} (|0\rangle + e^{i \phi_2} |1\rangle
\end{align}
where $\phi_1$ and $\phi_2$ satisfy the finite condition (\ref{finite}) and are therefore both rational. A simple calculation gives
\begin{align}
\label{superposed}
|\psi_1\rangle +|\psi_2 \rangle = A (\cos \frac{\phi_3}{2} |0\rangle + \sin \frac{\phi_3}{2} e^{i \phi_4} |1\rangle)
\end{align} 
where $A$ is a normalisation factor, and
\begin{align}
\cos^2\frac{\phi_3}{2}&=\frac{2}{3+\cos(\phi_1-\phi_2)}  \nonumber \\
\phi_4&= \frac{\phi_1+\phi_2}{2}
\end{align}
The interesting equation is the first one. By Niven's theorem described in Appendix C, if $\theta/2 \pi$ is rational, then generically $\cos \theta$ is not. By construction $\phi_1-\phi_2$ is rational. Hence $\cos^2 \phi_3/2$ cannot be rational. This means the normalised Hilbert state
$$
\cos \frac{\phi_3}{2} |0\rangle + \sin \frac{\phi_3}{2} e^{i \phi_4} |1\rangle
$$
cannot satisfy (\ref{finite}), no matter how large is $N$. This implies that superposed Hilbert states such as (\ref{superposed}) have no ontic status in IST. Essentially, the wavelike properties of quantum physics arise from the helical nature of trajectories on $I_U$, on which arithmetic closure is achieved using $p$-adic algebra \cite{Palmer:2018a}. 

\section*{Appendix C}
\begin{theorem}
 Let $\phi/\pi \in \mathbb{Q}$. Then $\cos \phi \notin \mathbb{Q}$ except when $\cos \phi =0, \pm \frac{1}{2}, \pm 1$. \cite{Niven, Jahnel:2005}
\end{theorem}
\begin{proof} 
Assume that $2\cos \phi = a/b$ where $a, b \in \mathbb{Z}, b \ne 0$ have no common factors.  Since $2\cos 2\phi = (2 \cos \phi)^2-2$, then $2\cos 2\phi = (a^2-2b^2)/b^2$. Now $a^2-2b^2$ and $b^2$ have no common factors, since if $p$ were a prime number dividing both, then $p|b^2 \implies p|b$ and $p|(a^2-2b^2) \implies p|a$, a contradiction. Hence if $b \ne \pm1$, then the denominators in $2 \cos \phi, 2 \cos 2\phi, 2 \cos 4\phi, 2 \cos 8\phi \dots$ get bigger without limit. On the other hand, if $\phi/\pi=m/n$ where $m, n \in \mathbb{Z}$ have no common factors, then the sequence $(2\cos 2^k \phi)_{k \in \mathbb{N}}$ admits at most $n$ values. Hence we have a contradiction. Hence $b=\pm 1$ and $\cos \phi =0, \pm\frac{1}{2}, \pm1$. 
\end{proof}
\bigskip
\bibliography{mybibliography}
\end{document}